\documentclass{article}
\usepackage{makeidx,amsmath,amssymb, latexsym, url}
\usepackage{epsfig}
\usepackage{fullpage}

\newtheorem{lemma}{Lemma}

\newtheorem{theorem}{Theorem}
\newenvironment{proof}
{{\noindent\bf Proof}}{$\Box$}

\title{Some New Equiprojective Polyhedra\thanks{This paper combines two conference papers~\cite{HHNL08}
and~\cite{RQH06}.}}

\author{Masud Hasan$^{1,2}$, Mohammad Monoar Hossain$^1$, Alejandro L\'opez-Ortiz$^2$, Sabrina Nusrat$^1$\\
Saad Altaful Quader$^1$, and Nabila Rahman$^1$\\
\\
$^1$ Department of Computer Science and Engineering\\
Bangladesh University of Engineering and Technology, Dhaka-1000, Bangladesh\\
masudhasan@cse.buet.ac.bd, monowar3306@yahoo.com, nusratsabrina@yahoo.com\\
saad0105050@yahoo.com, nabilar@gmail.com\\
\\
$^2$ Cheriton School of Computer Science\\
University of Waterloo, Waterloo, Ontario N2L 3G1, Canada\\
\{alopez-o, m2hasan\}@uwaterloo.ca}

\date{}

\begin{document}
\maketitle

\begin{abstract}
A convex polyhedron $P$ is $k$-equiprojective if all of its orthogonal projections,
i.e., shadows, except those parallel to the faces of $P$ are $k$-gon for 
some fixed value of $k$.
Since 1968 it is an open problem to construct all equiprojective polyhedra.
Recently, Hasan and Lubiw~\cite{HL08} have given a characterization of equiprojective polyhedra.
Based on their characterization, in this paper we discover some new equiprojective polyhedra
by cutting and gluing existing polyhedra.\\
\\
{\bf Keywords:} Algorithm, equiprojective polyhedra, orthogonal projection, prism, zonohedra.
\end{abstract}

\section{Introduction}
A (3D) \emph{convex polyhedron} is the region bounded by the 
intersection of a finite number of half-spaces. 
A convex polyhedron $P$ is {\it $k$-equiprojective\/} if 
its orthogonal projection, i.e., shadow, 
is a $k$-gon in every direction except directions parallel to faces of $P$.
A cube is 6-equiprojective, a triangular prism is 5-equiprojective
(in fact, any $p$-gonal prism is $p+2$-equiprojective), and
a tetrahedron is not equiprojective. 
See Figure~\ref{EQUIPROJECTIVE_FIGURE}(a).

\begin{figure}[htbp]
\begin{center}
\input{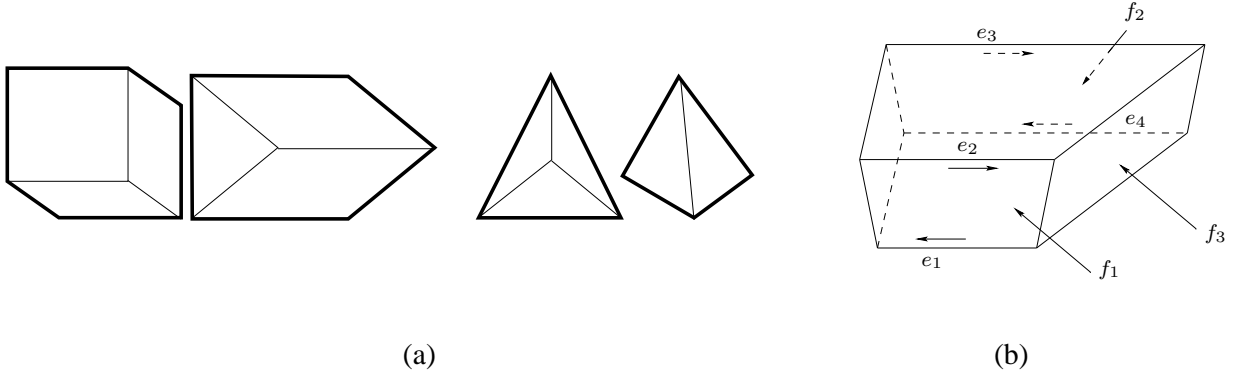}
\caption{
(a) A cube is 6-equiprojective, a triangular prism is 5-equiprojective, and a tetrahedron is not equiprojective.
(b) Examples of some compensating edge-face duples:
$(e_1,f_1)$ is compensated by $(e_2,f_1)$ and by $(e_4,f_2)$ but not by $(e_3,f_2)$.
$f_3$ is a self-compensating face.}
\label{EQUIPROJECTIVE_FIGURE}
\end{center}
\end{figure}

In 1968, Shepherd~\cite{She68,CFG91} defined equiprojective polyhedra, 
gave the examples above, and asked for a method to construct all equiprojective polyhedra.  
Recently, Hasan and Lubiw~\cite{HL08} have given a characterization 
and linear time recognition algorithm for equiprojective polyhedra. 
Using that characterization, 
in this paper, we prove that there is no 3- or 4-equiprojective polyhedron 
and triangular prism is the only 5-equiprojective polyhedron
(and thus the minimum equiprojective polyhedron) (Section~\ref{se:min}).
Then we give some new equiprojective polyhedra by applying 
some cutting and glueing techniques on existing convex polyhedra (Section~\ref{se:new}.

\section{Preliminaries}

Consider a convex polyhedron $P$.
For edge $e$ in face $f$ of $P$, we call $(e,f)$ an {\it edge-face duple\/}.
Two edge-face duples $(e,f)$ and $(e',f')$ are {\it parallel\/}
if $e$ is parallel to $e'$ and $f$ is parallel or equal to $f'$.
Define the {\it direction\/} of duple $(e,f)$ to be
a unit vector in the direction of edge $e$ as encountered 
in clockwise traversal of the outside of face $f$.

Two edge-face duples $(e,f)$ and $(e',f')$ {\it compensate\/} each other if
they are parallel and 
their directions are opposite (i.e., one is the negation of the other).
In particular, this means that either $f = f'$ and $e$ and $e'$ are
parallel edges lying at opposite ends of $f$,
or $f$ and $f'$ are distinct parallel faces and $e$ and $e'$ are
parallel edges lying at the ``same side'' of $f$ and $f'$,
where by ``same side'' it means that a plane through $e$ and $e'$
will have $f$ and $f'$ in the same half-space. 
An edge-face duple has at most two compensating duples.

A face $f$ is called \emph{self-compensating} when its edge-face duples
are compensated within themselves, e.g., the edges of $f$ are in parallel pairs.
See Figure~\ref{EQUIPROJECTIVE_FIGURE}(b).

\emph{Zonohedra} are the convex polyhedra where every face 
consists of parallel pairs of edges \cite{Tay92}.
(``Zonohedra'' has been defined by several mathematicians 
and there are some confusions in their definitions, 
see \cite{Tay92} for the history).
Zonohedra have the property that every face has a parallel 
face with corresponding edges parallel---each edge of a face has a 
parallel edge in the adjacent face, which in turn has another parallel
edge in the other face adjacent to the later, and so on; 
So a pair of parallel edges in a face 
will subsequently give two parallel edges in a parallel face.
For more information on zonohedra, see the web pages \cite{Epp,Epp96,Har}.

A \emph{$k$-gonal prism} consists of two similar copies of $k$-gons with 
each pair of corresponding edges connected by a parallelogram.


The characterization of equiprojective polyhedra given in \cite{HL03} 
is based on partitioning the edge-face duples into compensating pairs.
Any edge $e$ of the shadow of $P$ corresponds to some edge of $P$.
As the projection direction changes, $e$ may leave the shadow 
boundary.  This only happens when a face $f$ containing $e$ 
in $P$ becomes parallel to the projection direction.  
In order to preserve the size of the shadow, some other edge
$e'$ must join the shadow boundary.  
For these events to occur simultaneously, $e'$ must
be an edge of $f$, or of a face parallel to $f$.
This gives some idea that the condition for equiprojectivity
involves a pairing-up of parallel edge-face duples of $P$.

\begin{theorem}{\rm \cite{HL03}}
\label{th:char}
Polyhedron $P$ is equiprojective iff its set of edge-face 
duples can be partitioned into compensating pairs.
\end{theorem}

By the above theorem, 
a zonohedron is equiprojective since all of its faces are self-compensating;
a prism is equiprojective since the corresponding edges in two opposite bases
compensate each other and the remaining faces are self-compensating;
whereas a pyramid is not equiprojective since the edges in the triangular faces
do not have any compensating pairs (as the triangles do not have parallel pairs).  

\section{Minimality}
\label{se:min}

Observe that both a zonohedron and a prism have at least one pair of parallel faces.
In fact, in the following lemma we prove that this is true for any equiprojective polyhedron.

\begin{lemma}
\label{le:two_parallel}
Let $P$ be an equiprojective polyhedron. Then $P$ has at least one pair of parallel faces.
\end{lemma}

\begin{proof}
By way of contradiction assume that $P$ has no parallel faces.
Since $P$ is equiprojective, by Theorem~\ref{th:char} all of its edge-sides are 
partitioned into compensating pairs.
By definition, an edge-side can be compensated by another one within 
the same face or in the parallel face.
So in $P$ each edge-side is compensated within the same face,
which implies that all faces of $P$ are self-compensating,
further implying that $P$ is a zonohedron. 
But the faces of a zonohedron are also in parallel pairs, a contradiction.
\end{proof}

\begin{theorem}
There is no 3- or 4-equiprojective polyhedron, 
and a triangular prism is the only 5-equiprojective polyhedron.
\end{theorem}

\begin{proof}
Let $P$ be $k$-equiprojective for minimum possible value of $k$.
We first prove that $k\ge 5$.
By Lemma~\ref{le:two_parallel}, let $f$ and $f'$ be two parallel faces of $P$.
Consider a projection of $P$, and let $P'$ be the convex polygon of this 

projection.
The number of edges in $P'$ is $k$.
$P'$ contains edges from $f$ and $f'$ as two non-empty chains, one for $f$ and another for $f'$.
Let these two chains be $C$ and $C'$, respectively.
Since $P$ is convex, in $P'$, the two chains $C$ and $C'$ are disjoint and 
must be separated by at least two edges from faces other than $f$ and $f'$.
So the size of $P'$ is at least two plus the size of $C$ and $C'$.
Since $P$ is equiprojective, $f$ and $f'$ are self-compensating or 
compensate each other.
If $f$ and $f'$ are self compensating, then each of them has at least $2m$ edges, for some $m \ge 2$
with $m$ of those edges appearing in $p'$~\cite{HL08}, 
which gives the size of $P'$ at least six.
On the other hand, if $f$ and $f'$ compensate each other, then among two parallel edges
$e\in f$ and $e'\in f'$, exactly one appears in $P'$~\cite{HL08}.
Which gives that $|C|+|C'|=\frac{1}{2}|f|+|f'|$,
which is at least three when $f$ and $f'$ are triangles.
Therefore, $k\ge 5$.

Now, to keep the size of $P'$ as five, the only way to realize $P$ is to connect three edges of $f$ 
with the corresponding three parallel edges of $f'$ by three rectangles;
any other way to realize $P$ will have either more pair of parallel faces or 
a remaining face bigger than a rectangle.
In either case, $P$ will have a projection of size more than five.
\end{proof}

\section{The new polyhedra}
\label{se:new}
\subsection{New polyhedra from Johnson Solids}
In this section we give four new equiprojective polyhedra.
We derived them from four Johnson solids, which are a tetrahedron, a square pyramid, a triangular cupola, and a pentagonal rotunda,
and call them as \emph{equitruncated tetrahedron} (similarly, \emph{equitruncated pyramid}, and so on).
Our construction for equitruncated tetrahedron and equitruncated pyramid also work for the case when 
the tetrahedron and pyramid are not necessarily regular, and we will describe them for
an arbitrary tetrahedron and an an arbitrary pyramid.

\paragraph{Equitruncated tetrahedron}
We first see the construction of an equitruncated tetrahedron and then see the proof of its equiprojectivity.
Consider an arbitrary tetrahedron $P$.
We obtain an equitruncated tetrahedron $P'$ by cutting $P$.
Two types of cuts are applied to $P$: a \emph{top cut} and six \emph{side cuts}. 
(See Figure~\ref{fi:tetrahedra_my}(a).)
We consider an arbitrary face of $P$ as its \emph{base}
and the remaining vertex of $P$ as its \emph{apex}. 
The top cut is parallel to the base and it cuts the apex into a triangle.
The six side cuts can be divided into three similar pairs and are defined as follows.
Let the three edges adjacent to the apex of $P$ be $a,b$ and $c$.
Take three points $p_a, p_b,p_c$ on $a,b,c$ and imagine three parallel lines $l_a,l_b,l_c$ 
through $p_a, p_b,p_c$ such that the lines properly intersect the base
(e.g., they are perpendicular to the base, if $P$ is regular).
Take a pair of side cuts through $l_a$ and $l_b$ such that they are parallel to $c$.
Similarly, take another pair of side cuts through $l_a$ and $l_c$ that are parallel to $b$.
Take a third pair of side cuts through $l_b$ and $l_c$ that are parallel to to $a$.
This ends the construction of $P'$. See Figure~\ref{fi:tetrahedra_my}(a).

\begin{theorem}
\label{th:P1}
An equitruncated tetrahedron is 10-equiprojective.
\end{theorem}

\begin{proof}
Edges in the top triangle of $P'$ are compensated by the three (old) edges of the base.
Each pair of parallel cuts creates a pair of parallel triangles
which compensate each other.
Six (new) edges of the base are themselves in parallel pairs 
and thus compensate each other within the base. 
The remaining three hexagonal faces are self-compensating.
By Theorem~\ref{th:char}, $P'$ is equiprojective.
Finally, the non-degenerated orthogonal projection of $P'$ in Figure~\ref{fi:tetrahedra_my}(a) 
has size ten.
Therefore, an equitruncated tetrahedron is 10-equiprojective.
\end{proof}

\begin{figure}[htbp]
\begin{center}
\input{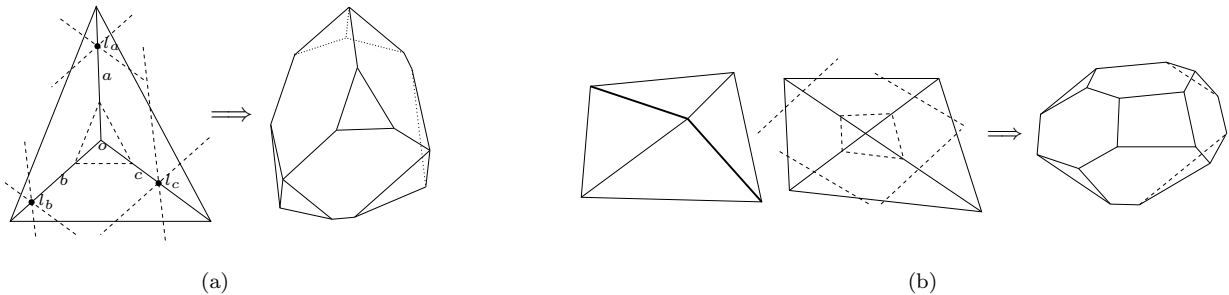}
\caption{(a) An equitruncated tetrahedron obtained from a tetrahedron.
(b) An equitruncated pyramid obtained from a pyramid.}
\label{fi:tetrahedra_my}
\end{center}
\end{figure}

\paragraph{Equitruncated pyramid}
An \emph{equitruncated pyramid} $P'$ is obtained by cutting an arbitrary quadrilateral pyramid $P$.
Again, two types of cuts are applied here: a \emph{top cut} and four \emph{side cuts}. 
(See Figure~\ref{fi:tetrahedra_my}(b).)
The top cut is parallel to the base and creates a quadrilateral at the top. 
For the side cuts, consider the four edges that are incident to the apex of $P$.
These four edges can be partitioned into two paths ${\cal P}_1$ and ${\cal P}_2$ 
where each path contains two edges from two different faces.
See Fig.~\ref{fi:tetrahedra_my}(b), where one of paths has been shown in bold.
Observe that the two edges of a path are coplanar. 
For each path, apply two side cuts at the two ends such that they are parallel 
to the plane of the other path.
This ends the construction of $P'$. 
See Figure~\ref{fi:tetrahedra_my}(b).

Proof of the equiprojectivity of $P'$ is similar to that of an equitruncated tetrahedron.

\begin{theorem}
\label{th:P2}
An equitruncated pyramid is 10-equiprojective.
\end{theorem}

\paragraph{Equitruncated triangular cupola}
We obtain an \emph{equitruncated triangular cupola} $P'$ from a triangular cupola $P$.
Before we go to the actual construction, let us review some properties of $P$.
See Figure~\ref{fi:j3_my}(a) for two complementary views of $P$.
Beside the regularity of its faces,
we observe that the edges of $P'$ that are not an edge of the hexagonal base 
can be partitioned into three similar paths, where each path 
contains three edges including one from the top triangle.
Moreover, edges in each path 
are coplanar and their plane is parallel with one of the three triangles 
that are adjacent to the base. 
See Fig.~\ref{fi:j3_my}(a), where one such path is shown in bold.

\begin{figure*}[htbp]
\begin{center}
\input{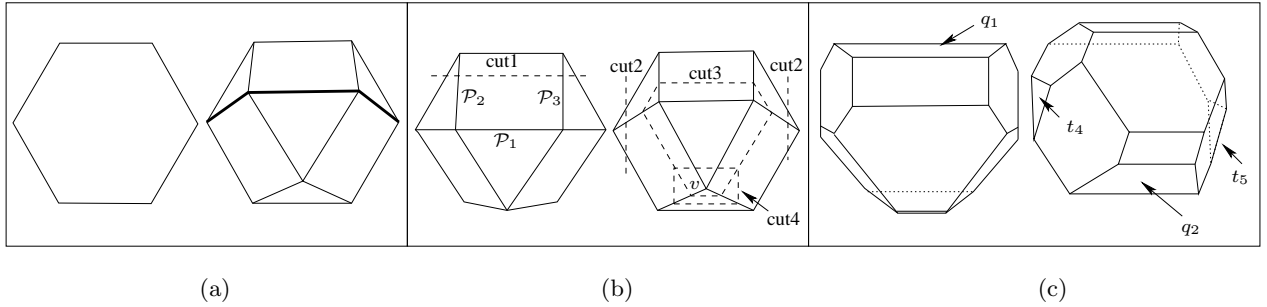}
\caption{(a) A triangular cupola $P$. A path is shown in bold.
(b) The four cuts applied to $P$. 
(c) The equitruncated triangular cupola obtained from $P$.}
\label{fi:j3_my}
\end{center}
\end{figure*}

To construct $P$ we use four types of cuts: \emph{cut1, cut2, cut3}, and \emph{cut4}.
See Figure~\ref{fi:j3_my}(b).
Let ${\cal P}_1, {\cal P}_2, {\cal P}_3$ be the three paths of $P$ (as mentioned above).
Let $t_1, t_2, t_3$ be three triangles adjacent to the base of $P$,
with $t_i$, for $i=1,2,3$, being parallel to the plane of ${\cal P}_i$. 
Cut1 is parallel to $t_1$, cuts two end edges of ${\cal P}_2$ and ${\cal P}_3$,
and creates a rhombus $q_1$.
Next, a parallel pair of cut2, which are perpendicular to the base and 
the plane of $t_1$, are applied to the two end edges of ${\cal P}_1$.
These two cuts create two parallel triangles $t_4$ and $t_5$.
Cut3, which is parallel to the base and the top triangle, 
is applied at the top and converts the top triangle into a hexagonal face.
Next, a cut4 is applied to the common vertex of ${\cal P}_2$ and ${\cal P}_3$
(we let this vertex to be $v$) such that it creates a rectangle. 
Observe that since the faces around $v$ are equilateral triangles and squares in alternate fashion
and since cut3 is parallel to the top triangle, such a rectangle is guaranteed by the cut4.
Cut4 also converts the triangle $t_1$ into a quadrilateral $q_2$ with two edges parallel.
See Fig.~\ref{fi:j3_my}(b,c).

\begin{theorem}
\label{th:P3}
An equitruncated triangular cupola is 11-equiprojective.
\end{theorem}

\begin{proof}
In $P$, all faces, except the triangles $t_4$ and $t_5$ and the quadrilaterals $q_1$ and $q_2$,
are self-compensating.
Edges of $t_4$ are compensated by the edges of $t_5$.
In each of $q_1$ and $q_2$, two edges are parallel and thus compensate each other
within the same face.
The remaining two edges of $q_1$ are parallel and thus compensated by the remaining two edges of $q_2$.
Finally, in  Figure~\ref{fi:j3_my} the non-degenerated orthogonal projection of $P$ has size eleven.
Therefore, an equitruncated triangular cupola is 11-equiprojective.
\end{proof}

\paragraph{Equitruncated pentagonal rotunda}
In this section we get two different \emph{equitruncated pentagonal rotunda} $P_1$ and $P_2$,
each from a pentagonal rotunda $P$.
As before, let us first study some properties of $P$.
Figure~\ref{fi:j6_my}(a) shows two opposite views of $P$.
Its base is a regular decagon and its other edges are partitioned into five paths, 
each of length five and each containing one edge from the top pentagon.
Each path lie in a single plane. 
Let $e_1,\ldots,e_5$ be the circular ordering of the edges of the top pentagon.
We denote the paths as ${\cal P}_1,\ldots,{\cal P}_5$ assuming that ${\cal P}_i$
contains $e_i$.
We call two paths ${\cal P}_i$ and ${\cal P}_j$ {\it non-adjacent} if $e_i$ and $e_j$
are nonadjacent. See Figure~\ref{fi:j6_my}(b).

To construct the first equitruncated pentagonal rotunda $P_1$,
four types of cuts are used: {\it cut1, cut2, cut3} and {\it cut4}.
See Figure \ref{fi:j6_my}(b).
A cut4 is applied at each vertex of $P$ that is not incident to the base of $P$
and it is applied in such a way that it creates a rectangle.
Let $e_1$ and $e_2$ be two parallel edges of the base decagon.
Let the two paths that ends at the four end points of $e_1$ and $e_2$ be ${\cal P}_1$ and ${\cal P}_3$.
The two other faces adjacent to $e_1$ and $e_2$ are a pentagon and a triangle,
and let them be $p$ and $t$ respectively.
A pair of cut2 and cut3 are applied simultaneously at the ends ${\cal P}_1$ and ${\cal P}_3$.
The cut2 is parallel to $p$, cuts $t$, and creates a new quadrilateral $q_1$.
On the other hand, the cut3 is parallel to $t$, cuts $p$, and creates a new quadrilateral $q_2$.

In a similar way, we apply another pair of a cut2 and a cut3 at the end edges of ${\cal P}_2$ and $P_4$.
Finally, a pair of cut1, which are perpendicular to the base and the plane of ${\cal P}_5$,
are applied to the two end edges of ${\cal P}_5$ and crate two parallel triangles.

The second equitruncated pentagonal rotunda $P_2$ is obtained similar to $P_1$,
except by replacing the second pair of cut2 and cut3 by two parallel pairs of cut1 
at the ends of ${\cal P}_2$ and ${\cal P}_3$, which create two pairs of parallel triangles.

\begin{figure*}[htbp]
\begin{center}
\input{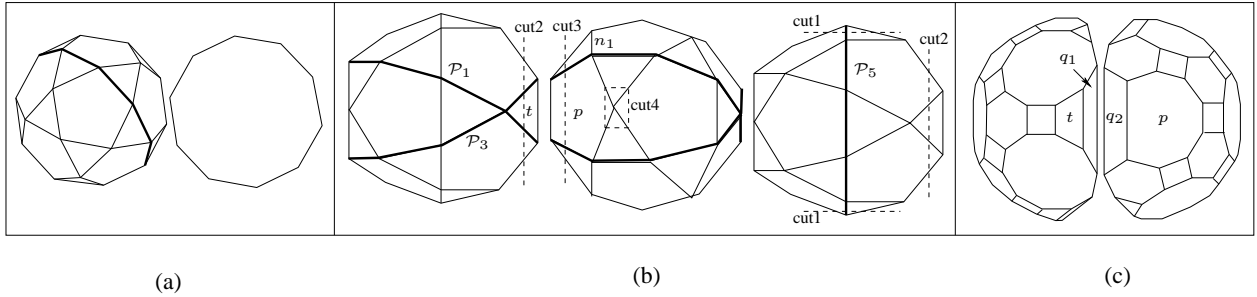}
\caption{Obtaining an equitruncated pentagonal rotunda $P_1$ from a pentagonal rotunda $P'$.
(a) Two opposite views of $P'$. (b) Cuts used in $P'$ to obtain $P_1$.
(c) Two complementary views of $P_1$.}
\label{fi:j6_my}
\end{center}
\end{figure*}

\begin{theorem}
The above two equitruncated pentagonal rotunda are $21$- and $23$-equiprojective respectively.
\end{theorem}

\begin{proof}
Let the two quadrilaterals created by the second pair of cut2 and cut3 be $q_1'$ and $q_2'$,
and let the second pair of pentagon and triangles parallel to them be $p'$ and $t'$ respectively.
For $P_1$, all faces except $p,t,p',t',q_1,q_2,q_1'$ and $q_2'$ are self-compensating.
For each of $p,p',t$ and $t'$, two edges are compensated by two edges of 
$q_1,q_1',q_2$ and $q_2'$, respectively.
All other edges of $p,p',t,t',q_1,q_1',q_2$ and $q_2'$ are compensated within the same face.
Proof for $P_2$ is similar.

For the equiprojectivity of $P_1$ to be $21$, consider the direction $d+\varepsilon$,
where $d$ is perpendicular to the base and does not see the base
and $\varepsilon$ is towards the outward normal of $t_1$.
For $\varepsilon$ small enough, the projection boundary of $P_1$ from 
$d+\varepsilon$ consists of six edges from $q_1,q_2$, another six edges from $q_1',q_2'$,
two edges from $t_2$ and seven edges from the base.
This makes a total of twenty one edges in the boundary of a non-degenerated projection of $P_1$.

For $P_2$, two pairs of parallel triangles replaces $q_1',q_2'$ of $P_1$.
As a whole, $P_2$ has three parallel pairs of triangles adjacent to the base.
So, for a projection of $P_2$ from a direction $d+\varepsilon+\varepsilon'+\varepsilon''$,
where $\varepsilon,\varepsilon',\varepsilon''$ are perpendicular to the
three planes of the triangles, respectively, 
the projection boundary consists of six edges from $q_1$ and $q_2$, 
six edges from the triangles, and eleven edges from the base.
This makes a total of twenty three edges in the boundary of a non-degenerated projection of $P_2$.
\end{proof}

\subsection{Polyhedra from zonohedra}

In this section we derive three new equiprojective polyhedra from
three zonohedra, \emph{rhombic dodecahedron} (henceforth called \emph{RD}),
\emph{truncated octahedron}, and \emph{truncated cuboctahedron (TC)}.
These polyhedra are shown in Figure \ref{fi:as}.
The basic idea is to cut each of them into two half and join them again in a slightly different way.

\begin{figure}[htbp]
\begin{center}
\input{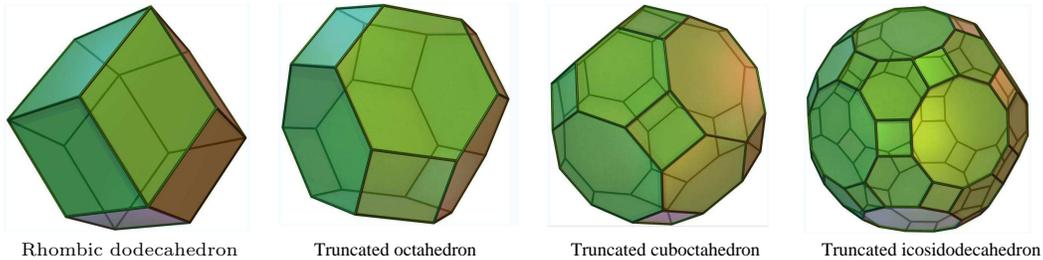}
\caption{Truncated octahedron, truncated cuboctahedron, truncated icosidodecahedron,
and rhombic dodecahedron (these figures are taken from Wikipedia: \texttt{http://www.wikipedia.com}).}
\label{fi:as}
\end{center}
\end{figure}

For an RD, we apply a cut through the diagonals of two opposite rhombus
and get two equal halves, which we call \emph{HRD}.
Observe that an HRD is equiprojective.
The new face that an HRD gets is called its \emph{base}.
We will join two HRD by their bases, but rotating one of them such that 
two triangular face adjacent to one HRD base become adjacent to rhombic 
face adjacent to the other HRD base.
We can not do it right away, because at this moment, the base of an HRD is not regular.
To make the base regular we apply an \emph{adjusting cut} which is parallel to bases
as shown in Fig.~\ref{fi:newj26}(a).
Observe that such a cut always exists, since the angles of the base of an HRD are equal.
We now join two such HRD as mentioned above and get a new equiprojective polyhedron,
which we call \emph{equitruncated rhombic dodecahedron}.
It is straight forward to see that the resulting polyhedron is convex and is 10-equiprojective.
See fig.~\ref{fi:as}(a).

Similarly, for a TO, we apply a cut through the diagonals of four squares 
and get two equal halves, which we call \emph{HTO}.
Once again, an HTO is equiprojective.
We join two HTOs by their bases after applying adjusting cuts 
and after rotating one of them such that triangular faces adjacent to the base of one HTO 
become adjacent to the hexagonal faces adjacent to the base of the other HTO,
and we get an equiprojective polyhedron, which we call \emph{equitruncated octahedron} and is 12-equiprojective.
See Fig.~\ref{fi:newj26}(b).

From a TC we can get three different equiprojective polyhedra, which we will call
as \emph{equitruncated cuboctahedron-I,II,III} respectively.
To achieve an equitruncated cuboctahedron-I, we apply a cut to TC that is parallel to an octagon 
and goes through an edge of four squares.
See Fig.~\ref{fi:newj26}.
This cut creates two polyhedra of unequal size.
The smaller one we call to be \emph{HTC}.
We take two copies of HTC and join them by their bases
after applying adjusting cuts and rotating one of them as before.
An equitruncated cuboctahedron-I is 13-equiprojective. 
Fig.~\ref{fi:newj26}(c).

\begin{figure}[htbp]
\begin{center}
\input{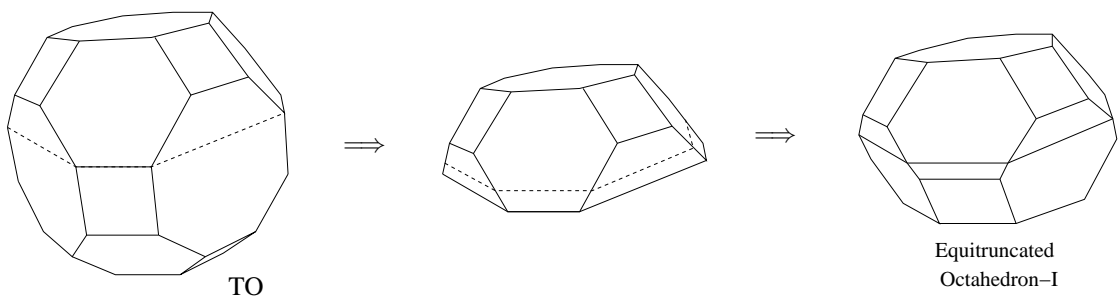}
\caption{Equiprojective polyhedra from an RD, a TC and a TO.}
\label{fi:newto}
\end{center}
\end{figure}

To achieve equitruncated cuboctahedron-II and III, first observe that among eight squares of a TO,
six are parallel to a direction.
We would like to apply a cut along the diagonals of six squares that are parallel to one direction.
However, such a cut does not exist.
To make such a cut possible, we will apply an adjusting cut parallel to each octagonal face 
as shown in Fig.~\ref{fi:new_4_6_8_ritated_2}.
The squares now turn to rectangles.
Moreover, we apply this adjusting cut such that the length of the diagonal of a rectangles becomes equal
to the length of an edge of the (enlarged) octagonal face.
Then we apply the cut through the diagonals of the six rectangles, 
which gives two equal halves of the TO, called \emph{HTO}.
Observe that the base of an HTO is a regular 12-gon.
We will rotate one HTO and join it to the other one by the bases such that 
the octagonal faces adjacent to one HTO base become adjacent to the similar octagonal faces of the other HTO.
This gives equitruncated cuboctahedron-II, which is 16-equiprojective.
To achieve  equitruncated cuboctahedron-III, we rotate one HTO and join to the other one 
such that octagonal and hexagonal faces that are adjacent to one HTO base
become adjacent to the adjacent triangles of the other HTO base.
This is a 17-equiprojective polyhedron.
See Fig.~\ref{fi:new_4_6_8_ritated_2}.

\begin{theorem}
An equitruncated octahedron, equitruncated cuboctahedron, and equitruncated cuboctahedron-I,II and III
are 10-, 12-, 13-, 16-, and 17-equiprojective, respectively.
\end{theorem}

It is possible to obtain equiprojective polyhedra from Truncated icosidodecahedron (see Fig.~\ref{fi:as})
in a way similar to obtaining truncated octahedron-II and III, but we find that too complicated
and too cumbersome.
We leave that as a future work.

\begin{figure}[htbp]
\begin{center}
\input{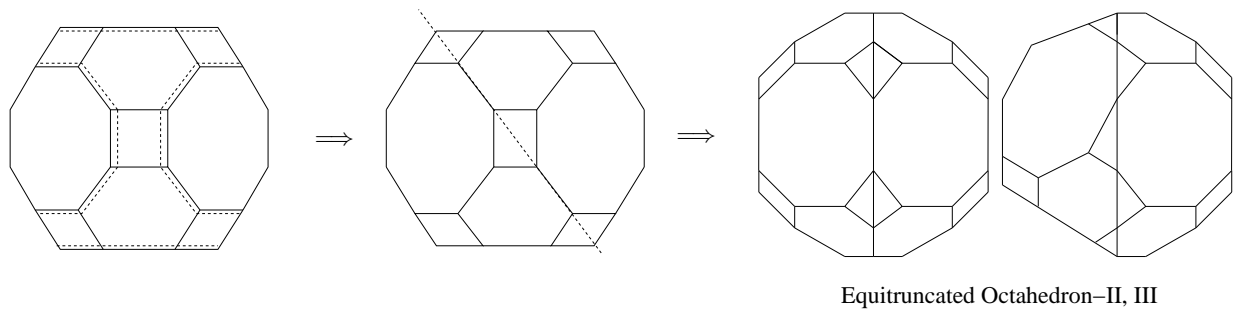}
\caption{Equiprojective polyhedra TO.}
\label{fi:new_4_6_8_ritated_2}
\end{center}
\end{figure}

\section{Joining prisms}
Our construction here is the generalization of the \emph{26th Johnson solid},
which is also called as \emph{gyrobifastigium}.
Gyrobifastigium is the join of two identical triangular prisms 
and is shown in Figure \ref{fi:newj26}(a).
Note that gyrobifastigium is 6-equiprojective.
The construction of gyrobifastigium can be generalized by generalizing the two prisms.
These two prisms can be $k_1$- and $k_2$-gonal, for arbitrary value of $k_1, k_2$.
The only criteria required is that the two faces by which the two prisms are joined  
should have solid angles at their edges such that after the joining 
the resulting polyhedron is convex. 
We call this generalized construction \emph{equiprojective bi-prism}.
Our construction is shown in Figure \ref{fi:newj26}(b), 
where the two prisms joined are triangular and 4-gonal respectively.

\begin{figure}[htbp]
\begin{center}
\input{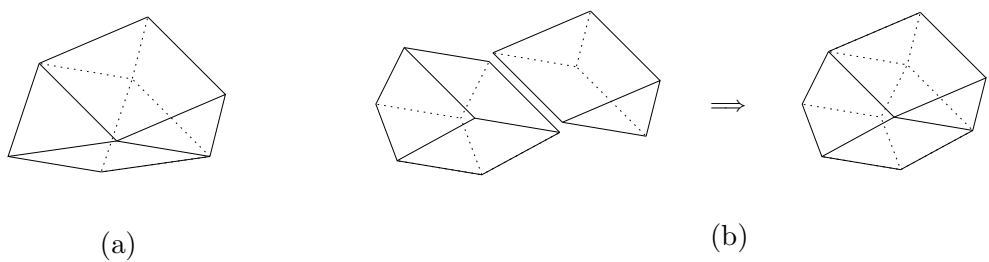}
\caption{Modified 26th Johnson solid by joining arbitrary prisms.}
\label{fi:newj26}
\end{center}
\end{figure}

The justification of why an equiprojective bi-prism is equiprojective is easy to follow.
Any prism is equiprojective and in the joining of two prisms 
we only lose two self-compensating faces.
Moreover, observe that an equiprojective bi-prism is $(k_1+k_2)$-equiprojective,
since the equiprojectivity of a gyrobifastigium is six and 
in an equiprojective bi-prism increasing the total size of the bases
by $k$ increases its equiprojectivity by $k$ too.

\begin{theorem}
Equiprojective bi-prisms are $(k_1+k_2)$-equiprojective, where the
two joining prisms are $k_1$- and $k_2$-gonal respectively.
\end{theorem}

\section{Conclusion}
We have implemented the above new equiprojective polyhedra in VRML.
The codes are available in~\cite{code10}, which can be downloaded for free
and the corresponding polyhedra can be viewed by using any VRML viewer.

In this paper we discovered some new equiprojective polyhedra.
We believe that the whole class of equiprojective polyhedra
is very rich, and so the actual open problem is still open: 
construct all equiprojective polyhedra.
We also believe that our technique of cutting and gluing existing polyhedra
may help discovering such an algorithm.
Any generalized algorithm for constructing even a subclass of equiprojective polyhedra
would be interesting and challenging.

\bibliographystyle{abbrv}
\bibliography{geombi}

\begin{thebibliography}{10}

\bibitem{CFG91}
C.~Croft, K.~Falconer, and R.~Guy.
\newblock {\em Unsolved Problems in Geometry}.
\newblock Springer-Verlag, New York, 1991.

\bibitem{Epp}
D.~Eppstein.
\newblock The geometry junkyard: Zonohedra.
\newblock \url{http://www.ics.uci.edu/~eppstein/junkyard/zono.html}.

\bibitem{Epp96}
D.~Eppstein.
\newblock Zonohedra and zonotopes.
\newblock {\em Mathematica in Education and Research}, 5(4):15--21, 1996.
\newblock \url{http://www.ics.uci.edu/~eppstein/pubs/Epp-TR-95-53.pdf}.

\bibitem{Har}
G.~W. Hart.
\newblock Encyclopedia of polyhedra.
\newblock \url{http://www.georgehart.com/virtual-polyhedra/vp.html}.

\bibitem{code10}
M.~Hasan, M.~M. Hossain, A.~L\'opez-Ortiz, S.~Nusrat, S.~A. Quader, and
  N.~Rahman.
\newblock {VRML} code for some new equiprojective polyhedra, 2010.
\newblock \url{http://teacher.buet.ac.bd/masudhasan/vrml_code_journal.html}.

\bibitem{HHNL08}
M.~Hasan, M.~M. Hossain, S.~Nusrat, and A.~L\'opez-Ortiz.
\newblock Smallest and some new equiprojective polyhedra.
\newblock In {\em 11th International Conference on Computer and Information
  Technology (ICCIT 2008)}, pages 459--464, Khulna, Bangladesh, December 2008.
  IEEE CS.

\bibitem{HL03}
M.~Hasan and A.~Lubiw.
\newblock Equiprojective polyhedra.
\newblock In {\em 15th Canadian Conference on Computational Geometry}, pages
  47--50, Halifax, Canada, August 2003.

\bibitem{HL08}
M.~Hasan and A.~Lubiw.
\newblock Equiprojective polyhedra.
\newblock {\em Computational Geometry: Theory and Applications},
  40(2):148--155, 2008.

\bibitem{RQH06}
N.~Rahman, S.~A. Quader, and M.~Hasan.
\newblock Some new equiprojective polyhedra.
\newblock In {\em 9th International Conference on Computer and Information
  Technology (ICCIT 2006)}, pages 34--38, Dhaka, Bangladesh, December 2006.

\bibitem{She68}
G.~Shephard.
\newblock Twenty problems on convex polyhedra---{II}.
\newblock {\em Math. Gaz.}, 52:359--367, 1968.

\bibitem{Tay92}
J.~Taylor.
\newblock Zonohedra and generalized zonohedra.
\newblock {\em American Mathematical Monthly}, 99(2):108--111, 1992.

\end{thebibliography}

\end{document}